\documentclass[11pt]{article}

\usepackage{microtype}
\usepackage{graphicx}
\usepackage{subfigure}
\usepackage{booktabs} % for professional tables

\usepackage{hyperref, color}

\usepackage[utf8]{inputenc} % allow utf-8 input
\usepackage[T1]{fontenc}    % use 8-bit T1 fonts
\usepackage{hyperref}       % hyperlinks
\usepackage{url}            % simple URL typesetting
\usepackage{booktabs}       % professional-quality tables
\usepackage{amsfonts, amsthm, amsmath}       % blackboard math symbols
\usepackage{nicefrac}       % compact symbols for 1/2, etc.
\usepackage{microtype}      % microtypography
\usepackage{wrapfig}
\usepackage[capitalize]{cleveref} % by siuon: apparently must be loaded before algorithm
\usepackage[ruled]{algorithm2e}

\newtheorem{theorem}{Theorem}[section]

\newtheorem{lemma}[theorem]{Lemma}

\newtheorem{problem}[theorem]{Problem}

\newtheorem{definition}[theorem]{Definition}

\newcommand{\R}{\mathbb{R}}

\newcommand{\E}{\mathbb{E}}

\newcommand{\rank}{\mathrm{rank}}

\newcommand{\eps}{\varepsilon}

\newcommand{\bv}{\mathbf{v}}

\newcommand{\by}{\mathbf{y}}

\newcommand{\f}{\mathbf{f}}
\newcommand{\X}{\mathbf{X}}

\newcommand{\bepsilon}{\boldsymbol{\epsilon}}

\newcommand{\Lcal}{\mathcal{L}}

\newcommand{\Gcal}{\mathcal{G}}

\newcommand{\Jcal}{\mathcal{J}}
\newcommand{\Dcal}{\mathcal{D}}
\newcommand{\Fcal}{\mathcal{F}}

\newcommand{\Rcal}{\mathcal{R}}

\newcommand{\Scal}{\mathcal{S}}

\newcommand{\Ical}{\mathcal{I}}
\newcommand{\Tcal}{\mathcal{T}}

\newcommand{\norm}[1]{\left\|#1\right\|}
\newcommand{\inner}[1]{\langle#1\rangle}

\usepackage{mathtools}
\newcommand{\defeq}{\vcentcolon=}

\newcommand{\mse}{\mathrm{MSE}}
\newcommand{\err}{\mathrm{err}}

\def\colorful{0}

\ifnum\colorful=1
\newcommand{\new}[1]{{\color{red} #1}}

\else
\newcommand{\new}[1]{{#1}}

\fi

\title{Efficient Algorithms for Multidimensional Segmented Regression\footnote{Authors are ordered alphabetically.}}

\author{
  Ilias Diakonikolas\thanks{Supported by NSF Award CCF-1652862 (CAREER) and a Sloan Research Fellowship. Part of this 
work was performed at the Simons Institute for the Theory of Computing during the program on Foundations of Data Science.} \\
  University of Wisconsin, Madison\\
  {\tt ilias@cs.wisc.edu} \\
  \and
  Jerry Li\\
   Microsoft Research AI\\
  {\tt jerrl@microsoft.com} \\
  \and
  Anastasia Voloshinov \\
  University of Southern California\\
 {\tt voloshana@gmail.com}
}

\begin{document}

\maketitle

\thispagestyle{empty}

\setcounter{page}{0}

\begin{abstract}
We study the fundamental problem of fixed design {\em multidimensional segmented regression}: 
Given noisy samples from a function $f$, promised to be piecewise linear on an unknown set of $k$ rectangles, 
we want to recover $f$ up to a desired accuracy in mean-squared error.
We provide the first sample and computationally efficient algorithm for this problem in any fixed dimension. 
Our algorithm relies on a simple iterative merging approach, which is novel in the multidimensional setting.
Our experimental evaluation on both synthetic and real datasets shows that our algorithm is competitive
and in some cases outperforms state-of-the-art heuristics.
\new{Code of our implementation is available at \url{https://github.com/avoloshinov/multidimensional-segmented-regression}. }
\end{abstract}

\section{Introduction} \label{sec:intro}

% We study the fundamental problem of multidimensional segmented liner regression \new{cite}.
% The main result of this paper is the first sample and computationally efficient
% algorithm for this problem in multiple dimensions.
% Before we describe our contributions in detail, we provide some background.

% \subsection{Background} \label{ssec:background}
The {\em regression} problem (see, e.g., ~\cite{MT77})  
is one of the prototypical statistical tasks.
In a (fixed design) regression problem, we are given a set of $n$ observations 
$(\mathbf{x}^{(i)}, y_i)$, where the $y_i$ are the dependent
variables and the $\mathbf{x}^{(i)}$ are the independent variables, and our
goal is to  model the relationship between them. The standard assumption
is that there is a simple function family $\cal F$ that models the underlying relation, 
and that the dependent observations are perturbed by random noise. 
More formally, we assume that there exists a known function family 
$\cal F$ such that for some $f \in {\cal F}$ we have 
\begin{equation}
\label{model}
y_i = f(\mathbf{x}^{(i)})+\eps_i \; ,
\end{equation} where the $\eps_i$ are i.i.d. sub-Gaussian 
random variables (see Section~\ref{sec:preliminaries} for formal definitions). The quality of an approximation is typically measured using
the Mean Squared Error (MSE). 

The textbook case that $f$ is  linear is fully understood: It is well-known that the least-squares
estimator is statistically and computationally efficient. The more general setting 
that $f$ is {\em non-linear}, but satisfies some well-defined structural properties, has also been extensively
investigated~\cite{GF73, feder1975, friedman1991, BP98, YP13, KRS15, ASW13, meyer2008, chatterjee2015} 
and is still an active research topic.
Indeed, the non-linear setting is not well-understood from an information-theoretic
and/or computational standpoint.

In this paper, we study the case that the function $f$ is promised 
to be {\em piecewise linear} with a given number $k$ of {\em unknown} $d$-dimensional rectangles. 
This is known as fixed design  {\em multidimensional segmented} regression, 
and has received considerable attention in the statistics community
~\cite{GF73, breiman1984classification,feder1975, quinlan1992learning,BP98, loh2002regression,hothorn2006unbiased,loh2011classification, YP13, ADLS16}. 
% The special case of piecewise polynomial functions (splines) has been extensively 
% used  in the context of inference, including density estimation and regression,
% see, e.g.,~\cite{WegW83, friedman1991, Stone94, Stone97, meyer2008}. 
Information-theoretic aspects of the segmented regression problem are well-understood:
Roughly speaking, the minimax risk is inversely proportional to the number of samples. 
In contrast, the computational complexity of the problem is poorly understood.
Known methods with provable guarantees, e.g., those presented in~\cite{blanchard2007optimal}, 
suffer worst-case runtimes of $\Omega(n^d)$, where $n$ is the number of data points. 
Moreover, their guarantees are often not sufficiently strong to actually recover the function $f$ 
in the traditional mean-squared-error metric (as we explain in Section~\ref{sec:our-results}).
In practice, heuristic methods such as CART~\cite{breiman1984classification} 
or GUIDE~\cite{loh2002regression} are often used, but to date there are 
no provable guarantees for the MSE of these estimators in this setting.
The CART algorithm in particular remains very popular in practice, 
and is the default implementation for regression trees in SciPy.

Many of these heuristics, including CART, allow the rectangles that determine $f$ 
to depend on all $d$ of the variables.
When $d$ is very large, the geometry of such trees becomes incredibly complex.
Indeed, it is straightforward to demonstrate that solving this problem efficiently would 
yield a polynomial time algorithm for PAC learning decision trees over $d$ variables with $k$ leaves. 
This is a notorious open problem in computational learning theory, 
believed to require at least $k^{\Omega(\log d)}$ time~\cite{EhrenfeuchtHaussler:89}.

To avoid this bottleneck, we consider a natural restriction of the general 
multidimensional segmented regression problem, where we assume that there is a known set 
$S$ of $d' \ll d$ coordinates so that the rectangles depend only on the coordinates in $S$.
That is, the position of these $d'$ coordinates at a data point $\mathbf{x}$ determine which linear fit applies to $\mathbf{x}$.
Such settings arise, e.g., in spatio-temporal datasets, where the linear predictor changes dramatically with time of year and/or location, but less so with other, secondary variables. When $d' = 1$, this problem reduces to the well-studied 
segmented regression problem~\cite{ADLS16}.
However,  for $d' > 1$, prior to this work, no computationally efficient algorithms 
with provable guarantees were known.

\subsection{Our Results}
\label{sec:our-results}
Our main contribution is the first computationally efficient algorithm, with provable performance guarantees, 
for multidimensional segmented regression in any fixed dimension $d'$.
Specifically, we give an algorithm \textsc{MultidimGreedyMerging}, satisfying the following:
\begin{theorem}[Informal, see Theorem \ref{mainthm}]
Let $f$ be a $k$-piecewise linear function over $\R^d$, 
where the rectangles that determine $f$ depend only on a known set of $d'$ variables, where $d' = O(1)$.
Given $\mathbf{x}^{(1)}, \ldots, \mathbf{x}^{(n)}$ and $y_1, \ldots, y_n$ generated by~\eqref{model}, 
where the noise $\eps_i$ is i.i.d sub-Gaussian, $\textsc{MultidimGreedyMerging}$ outputs $\hat{f}$ 
that with high probability satisfies
\[
\mse(\hat{f}) \defeq \frac{1}{n} \sum_{i = 1}^n (f(\mathbf{x}^{(i)}) - \hat{f}(\mathbf{x}^{(i)}))^2 = \widetilde{O} \left( \frac{k d}{n} + \sqrt{\frac{k}{n}} \right) \;.
\]
Moreover, the algorithm runs in time $\widetilde{O}(n d^2)$ time.
Here $\widetilde{O}(\cdot)$ hides polylogarithmic factors in its argument.
\end{theorem}
We make several remarks about the guarantee achieved by our algorithm.
First, it is folklore that the rate of $\Theta(kd / n)$ is minimax optimal for this estimation task.
Thus, when $d$ or $k$ is large in comparison to $n$, we match the minimax rate, up to logarithmic factors.

Second, our guarantee is for mean-squared error recovery of $f$, which 
is a strong notion of recovery. In particular, we note that mean-squared error recovery is stronger 
than other natural notions considered in prior work, including those in~\cite{blanchard2007optimal}.
As a result, these prior results do not have any implications for our setting.

Third, our algorithm runs in time that is {\em nearly-linear} in the number of data points $n$ 
and the number of rectangles $k$, for any constant $d'$.
Finally, we achieve this runtime by plugging in basic solvers for standard least-squares.
However, as we discuss later on, one can instead instantiate our solver with any least-squares solver, and our runtime will match it, up to polylogarithmic factors.
Thus, when $d'$ is constant, our runtime matches that of standard least-squares regression, up to polylogarithmic factors.

We validate the performance of our algorithm with experiments on both synthetic and real-world data.
We demonstrate that in reasonable settings, the performance of our algorithm compares favorably to CART, 
even when CART is allowed to branch on any coordinate, not just the ones in $S$.

\subsection{Our Techniques} \label{ssec:results}
In this section, we provide a brief overview of our algorithmic approach.
We start by observing that the algorithmic difficulty of the problem comes
from the fact that the location of the $k$ rectangles (in each of which $f$ is linear)
is unknown. For $d'=1$, there is a known, classical dynamic program (DP) that allows
us to ``find'' the unknown intervals (see, e.g.,~\cite{ADLS16}). Unfortunately, such a 
DP approach makes crucial use of the geometry of the univariate setting and does not 
generalize even to $d' = 2$. Roughly speaking, the $d'=1$ DP crucially uses the fact that 
merging two adjacent intervals creates another interval. However, in the multidimensional setting, 
the geometry is more complex (for example, merging two adjacent rectangles does not necessarily 
result in another rectangle) and DP seems to inherently fail. In summary, we are not aware
of any prior algorithm for this problem with provable runtime better than the brute-force bound 
of $n^{\Omega(d')}$.

Our algorithm uses an iterative greedy merging approach, generalizing an analogous
approach that has been used in the {\em univariate} setting~\cite{ADHLS15, ADLS16, ADLS17}.
The idea is to start from a large set of rectangles (defined by the input points) and iteratively
merge subsets of rectangles according to a judiciously chosen criterion.  We note that our iterative
merging approach is novel for the multivariate setting and we believe it will find further applications.
In recent work,~\cite{DLS18} employed an iterative {\em splitting} algorithm to perform density estimation
of multivariate histogram distributions. Our approach shares some features with~\cite{DLS18}. For 
example, we use a similar dyadic hierarchical partition of the space built on a data-dependent 
grid, which serves as the starting point of our algorithm. However, we emphasize that there are 
significant differences between our algorithm and its analysis, compared to~\cite{DLS18}.
\new{Perhaps the most notable difference is that our algorithm works ``bottom up'' 
as opposed to ``top down'' in~\cite{DLS18}. This makes both the algorithm and its analysis more subtle. As a result,
the accuracy guarantees we obtain are somewhat stronger than what would be achievable via a ``top down'' approach.} 

%Our approach is based on a greedy \textit{merging} scheme, as opposed to greedy splitting. 

%While the goal of the $d'=1$ setting was to improve the runtime of the dynamic program \cite{ADLS16}, we are not aware of a 
%polynomial time algorithm for our multidimensional problem. 

% \new{Techniques: Things we want to say
% \begin{itemize}
% \item The algorithm generalizes the iterative splitting method of~\cite{DLS18} in the context of density estimation.
% High-level idea is similar, but note differences in analysis.
% \item We develop an iterative merging method that is novel to the multidimensional setting.
% It has certain advantages and may lead to better error guarantees for other tasks, e.g., density estimation.
%\item In contrast to 1-d setting, DP doe not work in 2d. So, while the algorithm in 1-d has as a goal
% to improve the runtime (DP was already poly), we are not aware of a poly time algorithm (in constant dimension)
% for our multidimensional problem.
% \item We handle piecewise constant and linear regression (are we going to add kernels?). Classical method (CART) with provable
% convergence guarantees handles piecewise constant case. Our results can be viewed as a provable version of CART.
% \item Our rate of convergence is not quite minimax, but close. Nothing polynomial was known before, even for $d=2$.
% Our runtime is nearly linear in sample size.
%\item Experimental evaluation shows that we are competitive or outperform CART in various settings.
% \end{itemize}}

\section{Preliminaries and Background} \label{sec:preliminaries}
\subsection{Formal Problem Statement}
In this subsection, we formally define the problem of multidimensional segmented regression 
that we will study in this paper.

A {\em hyper-rectangle} (or rectangle for short) $R \subseteq [0, 1]^{d'}$ is  
a set of the form $R = \otimes_{i = 1}^{d'} I_i$, where each $I_i \subseteq [0, 1]$ is an interval.
For $\mathbf{x} \in [0, 1]^{d'} \times \R^{d - d'}$, we say that $\mathbf{x} \in R$, 
for a rectangle $R \subseteq [0,1]^{d'}$, if the first $d'$ coordinates of $\mathbf{x}$ lie within $R$.

We will consider a slightly generalized notion of piecewise linear functions, 
namely \emph{kernel piecewise linear functions}, and the corresponding regression 
problem of kernel segmented regression.
We let $\kappa: \R^d \to \R^{m}$ be a known, fixed kernel function.
When $\kappa$ is the identity map, this reduces to the normal notion of segmented regression.
This slight generalization will be helpful in the later experiments.
However, we encourage the reader to assume that $\kappa$ is the identity on first reading.

We now have the following definition.
\begin{definition}[$k$-piecewise linear functions]
Let $d \geq d'$.
We say that $f: [0,1]^{d'} \times \R^{d - d'} \to \R$ is a \emph{$k$-piecewise linear function} with kernel $\kappa$ if there exists a 
partition of $[0,1]^{d'}$ into $k$ axis-aligned hyper-rectangles $\Rcal^f = \{R_1^f,\ldots,R_k^f\}$ 
and vectors $\theta_1, \ldots, \theta_k$, such that $f(\mathbf{x}) = \langle \boldsymbol{\theta}_i, \kappa (\mathbf{x}) \rangle$ 
if $\mathbf{x} \in R_i^f$.
For a $k$-piecewise linear function $f$, we call $\Rcal^f$ its \emph{associated partition}.
\end{definition}
We note that the restriction that assumes that the first $d'$ coordinates are within $[0, 1]$ 
is without loss of generality, by scaling.

In this paper, we consider the \emph{fixed design} segmented regression problem.
We are given a fixed multiset of samples $\mathbf{x}^{(1)},\ldots,\mathbf{x}^{(n)} \in [0, 1]^{d'} \times \R^{d-d'} $, and we have some unknown $k$-piecewise linear function $f: [0, 1]^{d'} \times \R^{d-d'}$ with a known kernel $\kappa$.
We will measure error under the standard metric of mean squared error.
For any function $\tilde{f}: \R^d \to \R$, we define the mean squared error to be: 
$\mse(\tilde{f}) = \tfrac{1}{n} \sum_{i}^n  (\tilde{f}(\mathbf{x}^{(i)})-f(\mathbf{x}^{(i)}))^2$.

With this notation, we can now formally define our problem:

\begin{problem} 
Let $\mathbf{x}^{(1)},\ldots,\mathbf{x}^{(n)} \in [0, 1]^{d'} \times \R^{d-d'} $, and $f$ be as above.
Let $y_1, \ldots, y_n$ be generated by~\eqref{model}, where the $\eps_i$ are independent 
sub-Gaussian noise variables (see e.g.,~\cite{Rig15}), with variance proxy $\sigma^2$, 
mean $\E[\eps_i]=0$, and variance $s^2 = \E[\eps_i^2]$. 
Given $(y_1, \mathbf{x}^{(1)}), \ldots, (y_n, \mathbf{x}^{(n)})$, the goal is to output $\widetilde{f}$ 
minimizing $\mse(\widetilde{f})$.
\end{problem}

Note that by losing at most a factor of $2$, we may assume that $n$ is a power of $2$.

The following vector notation will also be useful shorthand later on.
We let $\bepsilon$ denote the vector of noise variables, that is, $\bepsilon_i = \eps_i$.
Similarly, let $\f$ denote the vector with components $\f_i = f(\mathbf{x}^{(i)})$ for $i \in [n]$. 
For any hyper-rectangle $R$, and any vector $\mathbf{v} \in \R^n$, we let $\mathbf{v}_R$ 
be the restriction of $\mathbf{v}$ to the coordinates $i$ so that $\mathbf{x}^{(i)} \in R$.

Finally, if $\tilde{f}$ is piecewise linear on some set of rectangles $\Rcal$, we define 
the error of $\tilde{f}$ on $R \in \mathcal R$ as $\err(R, \tilde{f}) :=  \|\tilde{\f}_R-\f_{R}\|_2^2$.
The $\mse$ can thus also be expressed as $\mse(\tilde{f}) = \tfrac{1}{n} \sum_{R \in \Rcal} \err(R, \tilde{f})$.

\subsection{Hierarchical Structure}
The true structure of the $k$ pieces of $f$ can be complicated, so as an intermediate step 
we introduce the notion of a hierarchical partition structure. 

Given $\mathbf{x}^{(1)}, \ldots, \mathbf{x}^{(n)}$, where $n$ is a power of $2$, we define an associated grid $\Gcal = P_1 \times P_2 \times \ldots \times P_d$, where $P_i = \{x^{(1)}_i, \ldots, x^{(n)}_i\} \subset [0, 1]$ is the collection of all the different $i$th coordinates in the dataset.
Let $v^{(1)}_i \leq v^{(2)}_i \leq \ldots \leq v^{(n)}_i$ be the elements of $P_i$ in sorted order.
With this notation, {\em the level-$\ell$ rectangles induced by $G$}, denoted by $R_\ell$, are defined to be 
$R_\ell= \{\otimes_{i=1}^d [v^{(i)}_{2^\ell j_i} , v^{(i)}_{2^\ell j_i + 1}] : j_i \in {0,\ldots, n/2^\ell - 1}\}$. 

The  dyadic decomposition with respect to a grid $\Gcal$, denoted $\Dcal = \Dcal(\Gcal)$, 
is defined to be $\Dcal = \cup_{\ell = 1}^{\log n} R_\ell$. We let $\Dcal_k$ denote all partitions of $\Dcal$ 
into $k$ disjoint rectangles. That is, the dyadic decomposition includes all of the axis-aligned 
rectangles created by continuously splitting the grid in half in each of the first $d'$ dimensions 
of the samples which define the grid.
A dyadic decomposition induces a natural complete $2^{d'}$-ary tree, 
where we think of the rectangle corresponding to the entire grid as the root, 
the rectangles in $R_{\log n}$ are the leaves, and a rectangle $R$ in level $\ell$ for $\ell = 1, \ldots, \log n - 1$ 
has edges to the rectangles $R'$ in level $\ell + 1$ so that $R' \subset R$.

We say that a function $f : [0,1]^{d'} \times \R^{d-d'} \rightarrow \R$ obeys a dyadic hierarchical partition 
with respect to a grid $\Gcal$, if there exists a partition of $[0,1]^{d'}$ into axis-aligned rectangles 
$R_1,\ldots,R_k \in \Dcal(\Gcal)$ so that $f$ is piecewise-linear in the first $d'$ coordinates on $R_i$.
Such a function naturally corresponds to a subtree of the complete tree described above.
Namely, we take smallest subtree of the complete tree so that $f$ is constant on the leaves of the subtree.
We will often refer to this as the tree \emph{associated} to $f$.
 
We first need the following lemma, which states that any partition with respect to a grid can be converted to a hierarchical partition with not too many more pieces.
\begin{lemma}
\label{lem:gen-to-hier}
Fix a grid $\Gcal$ with side length $n$. Let $f : [0,1]^{d'} \times \R^{d-d'} \rightarrow \R$ be a k-piecewise linear function that is piecewise in $d'$ dimensions, so that $f$ is constant on $R_1,\ldots,R_k$, and every vertex of every rectangle lies on $\Gcal$. Then $f$ obeys a $k \log ^{d'} n$-hierarchical partition. 
\end{lemma}
\begin{proof}
Any function which is supported within an axis-aligned rectangle R in $d'$ dimensions 
can be represented with a $\log^{d'}n$ hierarchical partition. 
Let $R=[a_1,b_1] \times [a_2,b_2] \times \ldots \times [a_{d'},b_{d'}]$. 
Every interval $[a_i, b_i]$ can be written as a union of at most $\log n$ disjoint dyadic intervals $\Ical_i$. 
So, $R$ can be decomposed as the disjoint union of all rectangles $R = \otimes_{i=1}^{d'} I_i$, 
where $I_i$ ranges over all intervals in $\Ical_i$. This requires $\log^{d'} n$ pieces. 
Since our function has $k$ rectangles, then it can be represented with $k \log^{d'} n$ hierarchical pieces.
\end{proof}

\subsection{Mathematical Preliminaries}
In this section, we state some mathematical preliminaries that our analysis uses.

%Our analysis requires the following tail bound on sub-exponential random variables:

%\begin{fact}[Bernstein-type inequality] \label{fact}
%Let $X_1,\ldots, X_N$ be centered sub-exponential variables, and let $K= \max_i \norm{X_i}_{\psi_1}$, the sub-exponential norm of $X$. This is the smallest parameter $K$ such that $\E(|X|^p)^{1/p} \leq Kp$ for all $p>1$. Then, for all $t>0$, we have

%\[\Pr \Big[  \Big| \frac{1}{\sqrt N} \sum_{i=1}^{N} X_i \Big| \geq t \Big] \leq 2 \exp \Big( -c \min \Big(\frac{t^2}{K^2}, \frac{t\sqrt N}{K}  \Big) \Big). \]
%\end{fact}

%The following lemma follows from Fact \ref{fact}:
We require the following bound on the noise:

\begin{lemma}\label{lem2} Fix $\delta>0$ and let $\eps_1,\ldots,\eps_n$ be as defined in (\ref{model}). With probability $1-\delta$, we have
\[
\Big| \sum_{i \in R} \eps_i^2 - s^2|R| \Big| \leq O(\sigma^2 \log(n/\delta))\sqrt{|R|} \; ,
\]
simultaneously, for all rectangles $R$ in the dyadic partition.
\end{lemma}

\begin{proof}
Let $\delta'=O(\delta/n^2)$. Let $\Tcal$ be the hierarchical tree induced by a $n \times n$ size grid $\Gcal$. 
Then, for any rectangle $R \in \Tcal$, we apply a Bernstein-type inequality (see, e.g., Theorem 1.13 in~\cite{Rig15}) 
to the sub-exponential random variable $X_i = \eps_i^2 - s^2$ for $i \in R$. This inequality depends 
on the sub-exponential norm of the random variable, which in this case is $K=\sigma^2$. 
From this inequality, we have that the desired bound holds with probability $1-\delta'$. 
By a union bound over all $O(n^2)$ rectangles in $\Tcal$, we get that desired bound holds with probability $1-\delta$, as claimed.
\end{proof}

\noindent
We next require the following lemma, which states that with high probability 
the random Gaussian noise is not too correlated with the function.
The proof follows from standard maximal inequalities, and we include it in an Appendix for completeness.
\begin{lemma}\label{cor4} 
Let \new{$m > 0$}. Let $\Lcal_{\new{m}}$ be the space of \new{$m$}-piecewise linear functions. With probability $1-\delta$, we have 
\new{
\[
\sup_{f \in \Lcal_{m}}\frac{|\langle \bepsilon_R, \f_R \rangle|}{\|\f_R\|_2} \leq O(\sigma \sqrt{m \cdot \rank(\kappa(\X)) +m \log(n/\delta)}) \; .
\]}
\end{lemma}

\noindent
With this in hand, we can prove the following guarantee for the error of the least squares fit, 
if we have identified rectangles on which the true function is linear.
The proof is very similar to the proof of Theorem 2.2 in~\cite{Rig15}, but we include it in an Appendix for completeness.
\begin{lemma} \label{lem1} 
\new{Let $\Rcal = \{R_1,\ldots,R_t\}$ be} such that $t = O(k)$, and let $f$ \new{be a piecewise linear function, so that it is a} linear function on each $R \in \Rcal$. 
\new{Let $\hat{f}$ be a $t$-piecewise linear function, so that on each $R \in \Rcal$, $\hat{f}$ is the linear least-squares fit to $f$ 
restricted to the points in $R$}.
Then, with probability $1-\delta$, we have
$\sum_{R \in \Rcal} \err(R, \hat{f}) \leq O(\sigma^2 k' (\rank(\kappa(\X))+\log(n/\delta)))$.
\end{lemma}

\section{Greedy Merging Algorithm} \label{sec:alg}

In this section, we present our algorithm for multidimensional segmented regression. 
Our algorithm begins by constructing a grid over the first $d'$ coordinates of the samples, 
so then each sample is located on a vertex of this grid. This grid induces a dyadic hierarchical partition. 
We view this partition as a hierarchical tree, where the root contains the entire grid, 
and the children split the parent into equal sized axis-aligned rectangles in the partition, 
as long as the children contain samples.

Our algorithm begins with a tree on the dyadic partition with $n$ leaf nodes 
and iteratively considers merging groups of sibling leaf nodes, 
which correspond to axis-aligned rectangles in the same level in the hierarchy. 
In each iteration, we fit a least square fit over each of the groups of sibling leaf nodes. 
We then merge all siblings except the $2k'$ groups that give us the largest regularized error measure. 
The regularized error measure is defined as 
\begin{equation}\label{eqn:reg-error}
\widetilde{\err}(R,\hat{f}_R)= \|y_R - \hat{f}_R\|_2^2 -\sigma^2|R| \;, 
\end{equation}
where $k'= k \log^{d'} n$. 
We repeat this process until we have less than $2k'$ groups of siblings up for consideration to be merged. 
The pseudo-code for our algorithm is given in Algorithm~\ref{alg:greedymerge} \new{and an illustration
is given in Figure~\ref{fig:partition}}.

\begin{figure}[!htb]
\centering
\includegraphics[width=\linewidth]{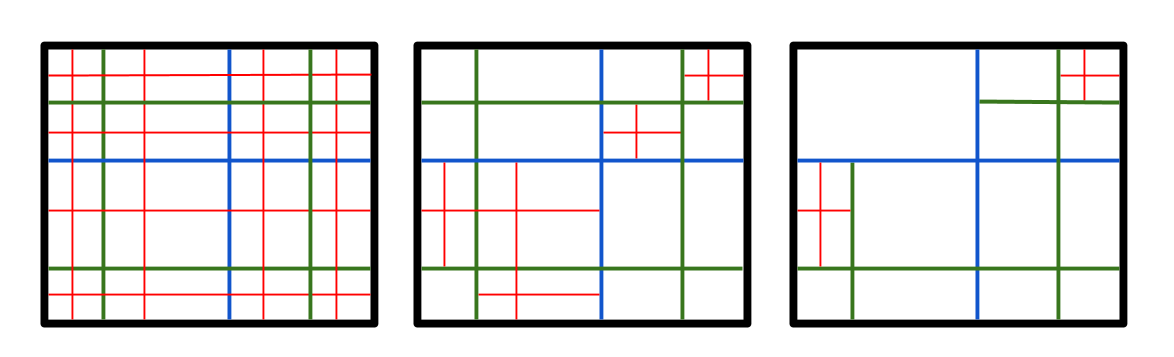}
\caption{Example of two iterations of the merging algorithm on the partitions. 
The left sub-figure displays the hierarchical partitioning of $\R^2$ in the beginning. 
The level-$1$ rectangles are bordered by the blue lines, 
the level-$2$ by the green, and the level-$3$ by the red. 
In the first iteration, the candidates for merging are the groups of $4$ rectangles bordered by red. 
The center figure shows the results after the merging has been completed --- the algorithm chooses 
to merge those that reduce \new{the regularized error of \eqref{eqn:reg-error}.}
The right-most figure shows another iteration of the algorithm.}
\label{fig:partition}
\end{figure}

The regularized error measure is used as a proxy for the true error, which we cannot measure. We do not merge together rectangles that give the largest regularized error measure, since this is an indication that these samples might not fit well in a piece. By not merging $2k'$ of the largest errors in each iteration, we have a guarantee that $k'$ of these were actually rectangles on which $f$ was flat (contained in a true piece of $f$), which will allow us to bound the error on the rectangles we merge. %This is a key idea that helps us bound the error on rectangles that we did merge in some iteration.

\begin{algorithm}[t]
\SetAlgoLined
\textsc{MultidimGreedyMerging}$(\X,\by)$ \\
Let $\Gcal$ be a grid over the first $d'$ coordinates of the samples in $\X$ \\
Let $\Tcal$ be a subtree of the hierarchical tree induced by $G$, initially containing all nodes which contain samples. \\
Let $\Scal$ be the collection of sets of sibling leaf nodes in $\Tcal$. \\
Let $k' = k \log^{d'} n$ \\
\While {$|\Scal| \geq 2k'$} {
\For {\text{each set of sibling leaves} $R \in \Scal$} {
Let $\hat{f}_R$ = \textsc{LeastSquares}($\kappa({\mathbf{X}_R}),\by_R$) \\
Let $\widetilde{\err}(R,\hat{f}_R)= \|y_R - \hat{f}_R\|_2^2 -\sigma^2|R|$
}
Let $\Jcal$ be the set of $2k'$ sibling sets $R \in \Tcal$ with largest $\widetilde{\err}(R,\hat{f}_R)$ \\
\For {\text{each} $R \notin \Jcal$} {
Merge the sibling leaf nodes together in $\Tcal$, so their parent becomes a leaf node
}
}
\Return The function which is the least squares fit for every leaf of $\Tcal$
\caption{Piecewise linear regression by greedy merging}
\label{alg:greedymerge}
\end{algorithm}

\subsection{Analysis of Algorithm~\ref{alg:greedymerge}}

In our algorithm, we use the blackbox subroutine $\textsc{LeastSquares}(\X,\by)$, where $\X$ is the $n \times d$ data matrix and $\by$ is the vector of labels. The classical algorithms for least squares that are commonly used in practice have time complexity $O(nd^2)$. We will assume this running time for this subroutine. So, when computing the least squares fit on some subrectangle $R$, $\textsc{LeastSquares}(\X_R,\by_R)$ runs in time $O(|R| \cdot d^2)$. 
With this, it is not hard to show the following runtime bound.
\begin{lemma}
\label{lem:runtime}
Algorithm~\ref{alg:greedymerge} runs in time $O(nd^2 \log n)$.
\end{lemma}
\begin{proof}
We go through a maximum of $\log n$ iterations. For each iteration, we need to call $\textsc{LeastSquares}$ 
for each of the groups of sibling leaves in $\Tcal$. For each group of leaves $S$, the runtime is $O(|S| \cdot d^2)$, 
and since the leaves are disjoint, the total runtime over all the groups is $O(n d^2)$. 
Thus, we get a runtime of $O(nd^2 \log n)$ over all iterations. 
\end{proof}

It is easily verified that by plugging in other solvers instead, we can also match their runtime, up to poly-logarithmic factors.

\noindent
The following simple lemma bounds from above the number of pieces that the algorithm produces.
\begin{lemma}
\label{lem:num-pieces}
Algorithm~\ref{alg:greedymerge} outputs a function that is piecewise \new{linear} on  $O(k'')$ pieces, where $k''=k \log^{d'+1} n$.
\end{lemma}
\begin{proof}
We stop merging if there are ever less than $2k'$ sibling leaf groups under consideration to be merged. 
Each of these groups is responsible for preventing at most $2^{d'}-1$ leaf nodes from being merged 
in each level on the path from them to the root node (since not all of the siblings of these leaves are also leaves). 
So, each group might block $2^{d'}\log n$ leaf nodes from merging. 
Therefore, a total of $(2^{d'}\log n)k' = 2^{d'}k''$ leaf nodes are blocked, where $k'' = k\log^{d'+1} n$. 
If we add $2^{d'+1} k'$nodes that were up for consideration but not merged, 
we then have $O(k'')$ total leaf nodes at the end.
\end{proof}

We are now ready to prove our main theorem.

\begin{theorem} \label{mainthm}
Let $\delta >0$ and let $\hat{f}$ be the estimator returned by $\textsc{MultidimGreedyMerging}$. 
Let $k' = O(k \log^{d'} n)$. Let  $k''=O(k\log^{d'+1} n)$ be the number of pieces in $\hat{f}$. 
Let $r = \rank(\kappa(\X))$.
Then, with probability $1-\delta$, we have
\[ \mse(\hat{f}) = O\bigg(\frac{\sigma^2k''(r+\log(n/\delta))}{n}+\frac{\sigma \sqrt{k'} \log(n/\delta) }{\sqrt{n}}\bigg) \;.\]
\end{theorem}
\begin{proof}
Let $\Rcal = \{R_1,\ldots,R_{k''}\}$ be the leaves output by the algorithm. 
We partition $\Rcal$ into two sets, and bound the error on the sets separately. 
We say $f$ is flat on a rectangle $R$ if $f$ over $R$ is defined by one linear function. 
We say that $f$ has a jump on $R$ if it is defined by more than one linear function over $R$. 
Let
$\mathcal F = \{R \in \Rcal: \text{f is flat on R} \}$ and 
$\mathcal J = \{R \in  \Rcal: \text{f has a jump on R}\}$.

We can bound the error over the rectangles in $\Fcal$ by directly applying Lemma \ref{lem1} to get $\sum_{R \in \Fcal} \err(R) \leq O(\sigma^2 k''(r+\log(n/\delta)))$.
Next we bound the error of the rectangles in $ \Jcal$. 
Consider some $R \in \Jcal$. If $|R|=1$, call this set $ \Jcal_1$. 
Then we know that $\hat{f}(x_i)=y_i$ for the $i\in R$. 
We get the following bound from Lemma~\ref{lem2}.
\begin{align*}
\sum_{R \in \Jcal_1} \|\f_R-\hat{\f}_R\|^2_2 &\leq \sum_{R \in  \Jcal_1} \|\bepsilon_R\|^2_2 \\
&\leq O \left(\sum_{R \in  \Jcal_1} |R|+\log (n/\delta)\sqrt{|R|} \right) \\
&\leq O  \left(\sigma^2 \left( k''+\log (n/\delta) \sqrt{k''} \right) \right) \; .
\end{align*}
Otherwise, $R \in \Jcal$ but $|R|>1$. 
Call this set $\Jcal_2$. So, for each $R \in \Jcal_2$, there was some iteration 
where there was a rectangle $R'$ such that $R' \subseteq R$, and $R'$ was merged in that iteration. 
Let the set of rectangles that were sub-rectangles of rectangles in $\Jcal_2$ 
and were merged at some iteration be $\Rcal'$. 

In an iteration where $R'$ was merged, there were $2k'$ rectangles, 
$ R_1,\dots,R_{2k'}$ such that $\widetilde{\err}(R') \leq \widetilde{\err}(R_j)$ for $j = 1,\ldots,2k'$. 
Out of these, we know that on at least $k'$ of them, $f$ must be flat. 
Let this set be $\Rcal^*$. We know that the error of each $R \in \Rcal^*$ 
can be bounded above by the average error of all $R \in \Rcal^*$. 
So, we have that 
$\widetilde{\err}(R') \leq \widetilde{\err}(R_j) \leq \tfrac{1}{k'}\sum_{R \in \Rcal^*} \widetilde{\err}(R)$.

We can bound $\sum_{R \in \Rcal^*} \widetilde{\err}(R)$ as follows
\begin{align}
\sum_{R \in \Rcal^*} \widetilde{\err}(R) &=
\sum_{R \in \Rcal^*} \|\by_R - \hat{\f}_R\|_2^2 -\sigma^2|R| \\
&= \sum_{R \in \Rcal^*} \|\f_R-\hat{\f}_R\|^2_2  \notag \\ &~~~~+ 2\sum_{R \in \Rcal^*} \langle \bepsilon_R, \f_R - \hat{\f}_R \rangle \notag \\
&~~~~+ \sum_{R \in \Rcal^*} \sum_{i \in R} (\eps_i^2 - \sigma^2) \label{eq:3terms}\\
&\leq O(\sigma^2 k' (r+\log(n/\delta))) \notag\\
 &~~~~+ O(\sigma \log (n/\delta)\sqrt{n}) \;. \label{eq:averagebound}
\end{align}
The first term in (\ref{eq:averagebound}) follows from bounding 
the first term of (\ref{eq:3terms}) with Lemma~\ref{lem1}, 
and the second term of (\ref{eq:3terms}) with Lemma~\ref{cor4}.
The second term of (\ref{eq:averagebound}) follows from Lemma~\ref{lem2}.

Thus, we divide by $k'$ to get that $\widetilde{\err}(R') \leq O(\sigma^2  (r+\log(n/\delta))) + O(\frac{1}{k'}\sigma \log (n/\delta)\sqrt{n})$.
Since we actually want to bound $\err(R')$, we bound from below $\widetilde{\err}(R')$:
\begin{align*}
\widetilde{\err}(R') &= \|\by_{R'}-\hat{\f}_{R'}\|^2_2-\sigma^2|R'| \\ &= 
\|\f_{R'}-\hat{\f}_{R'}\|^2_2  \\
&~~~~+ 2\langle \eps_{R'}, \f_R - \hat{\f}_R' \rangle  + (\|\eps_i\|^2_2 - \sigma^2|R'|) \\
&\geq \err(R') - O(\sigma \sqrt{r + \log(n/ \delta)})\|\f_{R'}-\hat{\f}_{R'}\|_2 \\ &~~~~- O(\sigma \log(n/\delta))\sqrt{|R'|} \; ,
\end{align*}
where the second term is bounded by Lemma \ref{cor4} and the last term is bounded by Lemma \ref{lem2}.

We combine this bound with (\ref{eq:3terms}) and rearrange to get 
\begin{align*}
\err(R') &\leq O(\sigma^2  (r + \log(n/\delta))) \\
&~~~~+ O(\sigma \sqrt{r + \log(n/ \delta)})\|\f_{R'}-\hat{\f}_{R'}\|_2\\
&~~~~+ O \left(\sigma \log(n/\delta)(\sqrt{|R'|} - \frac{\sqrt{n}}{k'}) \right)	\; .
\end{align*}
This inequality is of the form $z^2 \leq bz+c$, where $b,c>0$, so then $z^2 \leq O(b^2+c)$. 
Thus, we have
\begin{eqnarray*}
\err(R') \leq& O(\sigma^2  (r+\log(n/\delta))) + \\ 
&O \left(\sigma \log(n/\delta)(\sqrt{|R'|} + \frac{\sqrt{n}}{k'}) \right) \;.
\end{eqnarray*}
Therefore, the total error for rectangles in $\Jcal_2$ is 
\begin{align*}
&\sum_{R \in \Jcal_2}\err(R) \leq \sum_{R' \in \Rcal'}\err(R') \\ 
&\leq O(\sigma^2k'  (r+\log(n/\delta))) + O(\sigma \log(n/\delta))(\sqrt{k'n}) ,
\end{align*}
where the second term follows from the fact that the rectangles in $\Rcal'$ are disjoint.
Summing up the bounds we get for $\mathcal{F}, \mathcal{J}_1,$ and $\mathcal{J}_2$ completes the proof.
\end{proof}

\section{Experiments} \label{sec:experiments}

We study the performance of our new estimator for segmented regression on both synthetic and real data. All experiments were done on a laptop computer with a 2.5 GHz Intel Core i5 CPU and 8 GB of RAM. The focus of these evaluations was on statistical accuracy, not time efficiency. However, we note that the runtime of our algorithm was similar to that of CART. All algorithms took at most 18 seconds to run on the above computer architecture. For our synthetic data evaluations with piecewise constant true functions, our algorithm performs better in this measure. On real datasets, while our piecewise constant fits are worse than CART, we can perform better than CART if we use the full power of our estimator and output a piecewise linear predictor. \new{Code of our implementation and experiments is available at \url{https://github.com/avoloshinov/multidimensional-segmented-regression}. }

\paragraph{Synthetic data}
\begin{figure}[!htb]
\center{\includegraphics[width=80mm]{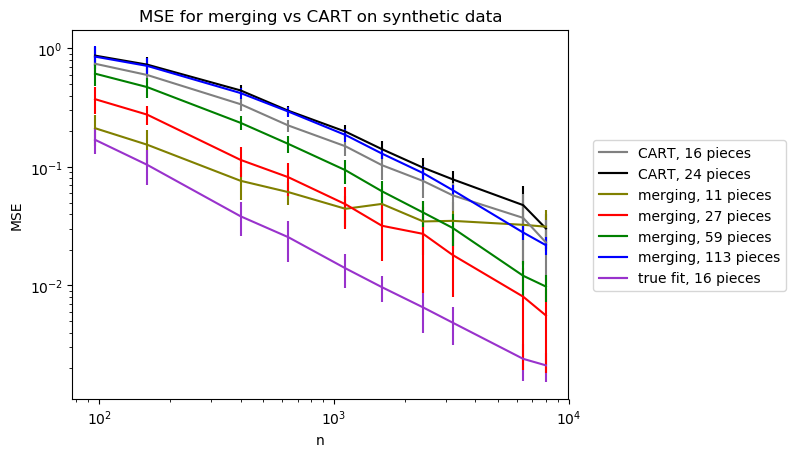}}
\caption{\label{fig:synthetic_mse} MSE of the merging algorithm and CART on synthetic data. There are four versions of our algorithm ``merging'' 
shown, with the average number of pieces that the algorithm produced over all of the trials over all of the values of $n$. There are two versions of CART shown --- one that was limited to producing $16$ pieces, and one that was limited to producing $24$ pieces. The ``true fit'' shows what the piecewise constant fit is on the true partition with $16$ pieces.}
\end{figure}

We first compare the statistical performance of our algorithm to CART on synthetic data. 
We used the ScikitLearn Julia library to import the DecisionTreeRegressor model from the Python skikit-learn library. 
This model implements CART (\url{https://scikit-learn.org/stable/modules/tree.html#tree}).

Since we are comparing to CART, which produces piecewise constant predictors, 
we consider the special case of our algorithm using constant predictors, to give the fairest comparison. 
Observe that this corresponds to the special case of the constant kernel $\kappa(\textbf{x}) = 1$.
We generate a function $f$ that is piecewise constant in $d'=2$ dimensions with a total of $d=10$ features. 
To generate the data, we draw $n$ (ranging from $n=96$ to $n=8000$) samples, 
where each coordinate is a normally-distributed random number with mean $0$ and standard deviation $1$. 
We then generate a piecewise constant function with $k=16$ pieces in $d'=2$ dimensions, 
by uniformly partitioning the data in the first two coordinates, 
such that each piece contains $n/k$ samples. Then, we pick a constant function for each piece, 
independently and uniformly at random from the interval $[0,1]$. 
We add i.i.d Gaussian noise with variance $1$ to each sample.

Figure \ref{fig:synthetic_mse} shows the average MSE over $20$ trials. 
The ``true fit'' shows the error of fitting a constant function on each of the true pieces. 
We ran CART with $16$ as the maximum number of leaves, as well as $24$ as the maximum number of leaves. 
We ran our algorithm ``merging'' with four different parameter settings, 
which resulted in an average of $11$, $27$, $59$, and $113$ pieces, 
for parameter settings respectively of $k, k/2, k/4,$ and $k/8$ 
for the number of candidate sets left when we stop merging. 
In theory, this parameter should be $2k'$, where $k'=k \log ^2 n$, but in practice, 
setting this parameter to smaller values and allowing our tree to keep merging works better, to a certain point. 
Most of our parameter settings achieved lower error than both of the CART algorithms for all values of $n$. 

\paragraph{Real data}
We investigate how our algorithm performs on real data through the 
Boston dataset (\url{https://www.cs.toronto.edu/~delve/data/boston/bostonDetail.html}). 
This dataset consists of $506$ samples, where each sample has $14$ attributes --- we use the first $13$ 
as features and the last as the label. The goal is to model the median value of owner-occupied homes in $1000$s 
of dollars. We chose this dataset because it is presented as the main example in the documentation for 
CART in scikit-learn (\url{https://scikit-learn.org/stable/modules/generated/sklearn.tree.DecisionTreeRegressor.html}), 
as well as in many other examples of CART. 

First, we compare the performance of CART on this dataset with a piecewise constant version of our algorithm 
(i.e., using the constant kernel). We run our algorithm and use the output of the number of pieces (25) 
as the input for how many pieces we want CART to output. For the merging algorithm, we use $4$ as the stopping parameter 
for merging, and $4$ as the value for sigma.  We compute the model based on all of the samples, 
and then look at the MSE of the model on all of the samples. 

For the stopping parameter, we tried values of $1, 2, 3, 4, 5, 6$, 
which resulted in piecewise fits ranging from $7$ pieces to $54$ pieces.
The choice for this parameter depends on the desired succinctness of the model. 
Since the comparisons to CART are similar for different value of this parameter, 
we just show the results for a single parameter. With a fixed stopping parameter (4), 
we tried $1, 2, 3, 4, 5, 10$ as values for sigma, representing the variance of the noise of the data. 
We used these parameters and ran our algorithm on the data, 
then used the value that gave the best MSE on the data. 
We note that as a result of our choice of \new{$\sigma$}, 
the MSE only differed by a maximum of $7$, and usually by only $1$-$2$.

\begin{figure}[!htb]
  \centering
  \begin{minipage}[b]{0.45\textwidth}
    \includegraphics[width=\textwidth]{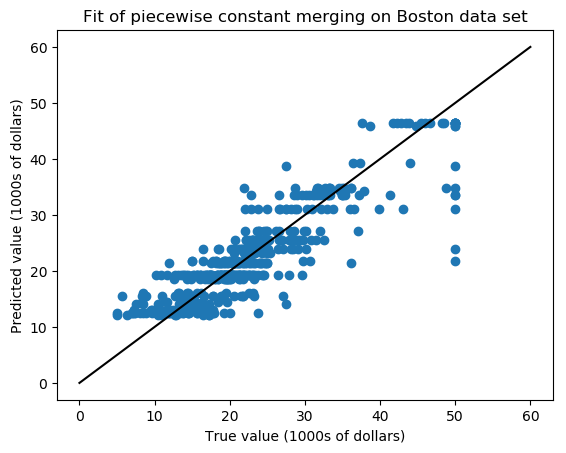}
  \end{minipage}
  \hfill
  \begin{minipage}[b]{0.45\textwidth}
    \includegraphics[width=\textwidth]{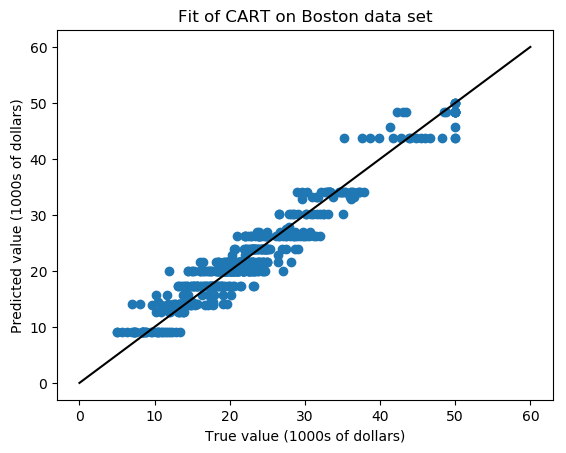}
  \end{minipage}
  \caption{The piecewise constant fit with $25$ pieces with the merging algorithm and CART. 
  The merging algorithm split on $2$ features: lstat (percent lower status of the population) and rm (average number of rooms per dwelling). 
  The MSE of the merging algorithm was $19.242$ and the MSE of CART was $6.155$.}
  \label{constant}
\end{figure}

Now, we look at the performance of CART on this dataset with the piecewise linear version of our algorithm 
(i.e., the identity kernel function $\kappa(\textbf{x}) = \textbf{x}$).  
Similarly to the constant experiment, we first run our algorithm, and use the output of the number of pieces 
as the input for how many pieces we want CART to output. For the merging algorithm, 
we use $3$ as the stopping parameter for merging, and $2$ as the value for sigma, 
which were chosen in the same manner as before. We compute the model based on all of the samples, 
and then look at the MSE of the model on all of the samples.

While our piecewise constant algorithm produced a result with worse MSE than CART, 
we can see that using our linear predictor can produce results with better MSE than CART in multiple regimes. 
We also note that our linear predictor, with sigma set between $1$ and $3$, outperforms CART 
for all stopping parameters that we chose, resulting in piecewise outputs on $7$ to $54$ pieces.

\begin{figure}[!htb]
  \centering
  \begin{minipage}[b]{0.45\textwidth}
    \includegraphics[width=\textwidth]{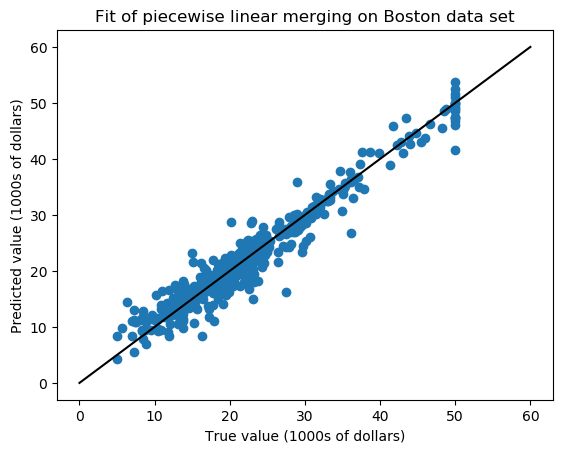}
  \end{minipage}
  \hfill
  \begin{minipage}[b]{0.45\textwidth}
    \includegraphics[width=\textwidth]{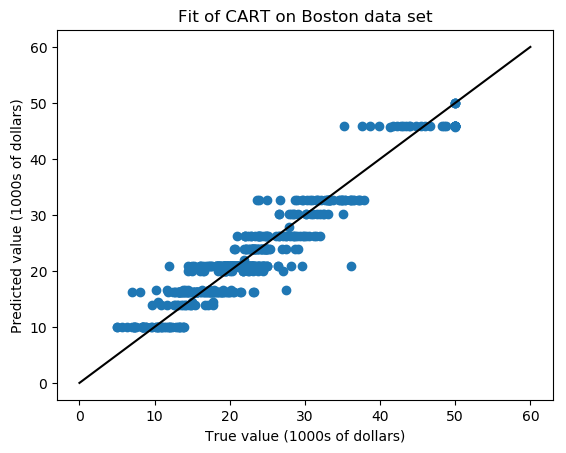}
  \end{minipage}
  \caption{The linear constant fit with $16$ pieces with the merging algorithm, and a $16$ piece constant fit with CART. 
  The merging algorithm split on $2$ features: lstat (percent lower status of the population) and rm (average number of rooms per dwelling). 
  The MSE of the merging algorithm was $5.464$ and the MSE of CART was $8.615$.}
  \label{linear_k2}
\end{figure}

\begin{figure}[!htb]
  \centering
  \begin{minipage}[b]{0.45\textwidth}
    \includegraphics[width=\textwidth]{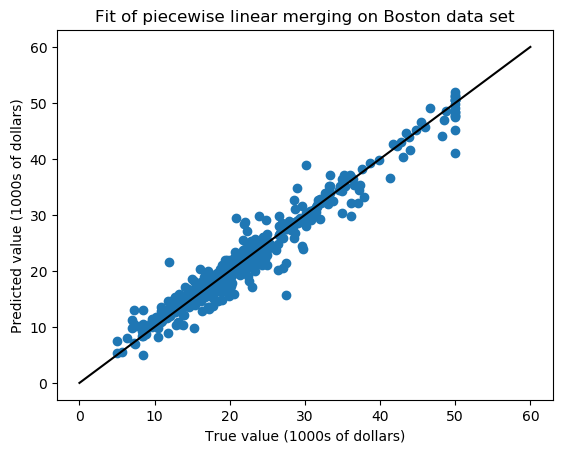}
  \end{minipage}
  \hfill
  \begin{minipage}[b]{0.45\textwidth}
    \includegraphics[width=\textwidth]{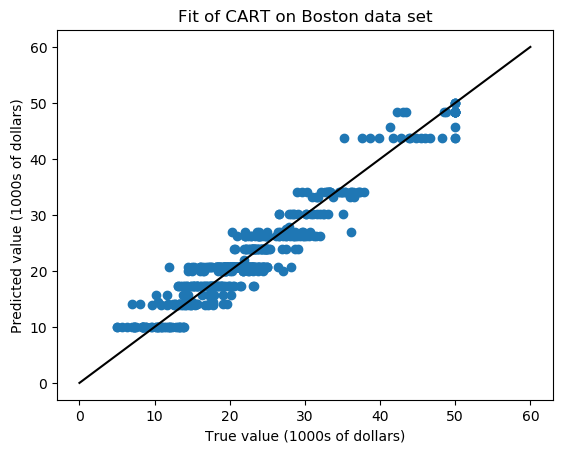}
  \end{minipage}
  \caption{The linear constant fit with $22$ pieces with the merging algorithm, and a $22$ piece constant fit with CART.
  The merging algorithm split on $3$ features: lstat (percent lower status of the population), rm (average number of rooms per dwelling), and dis (weighted distances to five Boston employment centers). The MSE of the merging algorithm was $4.303$ and the MSE of CART was $6.779$.}
  \label{linear_k3}
\end{figure}

\clearpage

\bibliography{allrefs}
\bibliographystyle{alpha}

\appendix

\section{Omitted Details from Section~\ref{sec:preliminaries}}

\subsection{Proof of Lemma~\ref{cor4}}
Before we prove the lemma, we need the following maximal inequality,
which bounds the correlation of a random vector with any fixed $d$-dimensional subspace, 
and the corollary bounds the correlation between sub-Gaussian random noise 
and any linear function on any rectangle. 

\begin{lemma}[see e.g., proof of Theorem 2.2 in~\cite{Rig15}]\label{lem4}
Fix $\delta>0$ and $\bv \in \R^n$. Let $\eps_1,\ldots,\eps_n$ be as defined in (\ref{model}). 
Let $\bepsilon = (\eps_1,\ldots,\eps_n)$, and let $S$ be a fixed, $r$-dimensional affine 
subspace of $\R^n$. Then, with probability $1-\delta$, we have 
\[\sup_{v \in S\setminus \{ 0\}} \frac{|\langle \bepsilon,\bv \rangle|}{\| \bv \|_2} \leq  O(\sigma \sqrt{r+\log(1/\delta)}) \;.\]
\end{lemma}
\noindent
With this lemma in hand we can now prove Lemma~\ref{cor4},

\begin{proof}[Proof of Lemma~\ref{cor4}]
Fix a partition of $[0,1]^2$ into $k'$ rectangles $\Rcal$, 
where each $R \in \Rcal$ is such that $R \in \Tcal$. 
Let $S_\Rcal$ be the set of $k'$-piecewise linear functions, 
which are linear fits on each $R \in \Rcal$. Then, $S_\Rcal$ is a $k' \cdot \rank(\kappa(\X))$-dimensional affine subspace. 
By Lemma \ref{lem4},
\[\sup_{f \in S_{\Rcal}}\frac{|\langle \bepsilon_R, \f_R \rangle|}{\|\f_R\|_2} \leq O(\sigma \sqrt{k' \rank(\kappa(\X)) +\log(1/\delta')}) \;,\]
with probability $1-\delta'$. The number of possible partitions $\Rcal$ is bounded above by ${n^2 \choose {k'}}=O(n^{2 k'})$. 
Let $\delta'=\delta/n^{2k'}$, then the result follows from a union bound over all possible partitions. 
\end{proof}

\subsection{Proof of Lemma~\ref{lem1}}

By the definition of the least squares fit, we have that 
$\norm{\by_\Rcal-\hat{\f}_\Rcal}^2_2 \leq \norm{\by_\Rcal-\f_\Rcal}^2_2 = \norm{\bepsilon_\Rcal}^2_2.$ 
If we expand the left hand side we get 
\begin{equation}
\norm{\by_\Rcal-\hat{\f}_\Rcal}^2_2  = \norm{\f_\Rcal+\bepsilon_\Rcal-\hat{\f}_\Rcal}^2_2  
=  \norm{\hat{\f}_\Rcal-\f_\Rcal}^2_2+ 2 \inner{\bepsilon_\Rcal, \f_\Rcal-\hat{\f}_\Rcal} + \norm{\bepsilon_\Rcal}^2_2 \;. \label{eq:expand}
\end{equation}
Applying Lemma \ref{cor4} gives us that with probability $1-\delta$, 
\begin{align*}
\norm{\hat{\f}_\Rcal-\f_\Rcal}^2_2 &\leq  2 \inner{\bepsilon_\Rcal, \hat{\f}_\Rcal- \f_\Rcal} 
\\ &\leq O(\sigma \sqrt{k'\cdot \rank(\kappa(\X)) +k'\log(n/\delta)}) \norm{\hat{\f}_\Rcal-\f_\Rcal}_2 \;.
\end{align*}
Rearranging this, we get that $\norm{\hat{\f}_\Rcal-\f_\Rcal}^2_2  \leq O(\sigma^2 k'\rank(\kappa(\X)+k'\log(n/\delta))$, 
which is what we wanted to show.

\end{document}